\documentclass[12pt,leqno,fleqn]{amsart}

\usepackage{amssymb,amsmath,amsthm,amscd,mathrsfs,amsbsy,amsfonts}
\usepackage{pstricks,pst-node}
\usepackage{amsbsy,young,colortbl,cite}

\usepackage{graphicx}

\usepackage{amsmath}

\usepackage{amssymb,amsthm}
\usepackage[english]{babel}
\renewcommand{\d}{\partial}

\newtheorem{proposition}{Proposition}
\newtheorem{lemma}{Lemma}
\newtheorem{definition}{Definition}
\newtheorem{theorem}{Theorem}

\makeatletter      % '@' is now a normail "letter" for TeX
\@addtoreset{equation}{section}

\newcommand{\gl}{M_N(\C)}

\newcommand{\I}{\mathbb{I}}
\newcommand{\bS}{\mathbb{S}}
\renewcommand{\d}{\mathrm{d}}

\newcommand{\Exp}[1]{\operatorname{e}^{#1}}

\newcommand{\g}{\mathfrak{g}}

\renewcommand{\L}{\mathcal L}

\newcommand{\W}{\mathcal W}

\newcommand{\R}{\mathbb R}
\newcommand{\Z}{\mathbb Z}
\newcommand{\C}{\mathbb C}
\newcommand{\N}{\mathbb N}
\def\res{\mathop{\rm Res}\nolimits}

\renewcommand{\L}{\mathcal{L}}

\renewcommand{\t}{\mathbf{t}}

\makeatletter
    
    \newcommand{\Rmnum}[1]{\expandafter\@slowromancap\romannumeral #1@}
  \makeatother
\def\res{\mathop{\rm Res}\nolimits}
\def\({\left(}
\def\){\right)}
\def\[{\begin{eqnarray}}
\def\]{\end{eqnarray}}
\def\d{\partial}
\def\ga{\alpha}

\newcommand{\La}{\Lambda}
\begin{document}

\title{The extended $Z_N$-Toda hierarchy}

\author{
Chuanzhong Li\dag,\  \ Jingsong He\ddag} \dedicatory {  Department of Mathematics and Ningbo Collabrative Innovation Center of Nonlinear Harzard System of Ocean and Atmosphere,\\
 Ningbo university, Ningbo 315211, China,\\
\dag lichuanzhong@nbu.edu.cn\\
\ddag hejingsong@nbu.edu.cn}
\thanks{}
\date{}

\maketitle

\begin{abstract}
The extended flow equations of a new $Z_N$-Toda hierarchy which takes values in a commutative subalgebra $Z_N$ of $gl(N,\mathbb C)$ is constructed. Meanwhile we give the Hirota bilinear equations and tau function of this new extended $Z_N$-Toda hierarchy(EZTH). Because of logarithm terms, some extended Vertex operators are constructed in generalized Hirota bilinear equations which might be useful in topological field theory and Gromov-Witten theory. Meanwhile the Darboux transformation and bi-hamiltonian structure of this hierarchy are given. From  hamiltonian tau symmetry, we give another different tau function of this hierarchy with some unknown mysterious connections with the one defined from
the point of Sato theory.
\end{abstract}


\tableofcontents

\section {Introduction}

The KP hierarchy and Toda lattice hierarchy
  as  completely integrable systems  have many important applications in mathematics and physics including the theory of Lie algebra representation, orthogonal polynomials and  random
matrix model  \cite{Toda,Todabook,UT,witten,dubrovin}. KP and Toda systems have many kinds of reduction or extension, for example BKP, CKP hierarchy, extended Toda hierarchy (ETH)\cite{CDZ,M}, bigraded Toda hierarchy (BTH)\cite{C}-\cite{ourBlock} and so on. There are another kinds of generalization called multi-component KP \cite{kac,avanM} or multi-component Toda systems which attract more and more attention because its widely use in many fields such as the fields of multiple orthogonal
polynomials and non-intersecting Brownian motions.

The multicomponent KP hierarchy  was discussed
with application on representation theory and random matrix model in \cite{kac,avanM}. In
\cite{UT},  it was noticed that $\tau$ functions of a
$2N$-multicomponent KP provide solutions of the $N$-multicomponent
2D Toda hierarchy.
The multicomponent 2D Toda hierarchy
 was considered from the point of view of the Gauss-Borel factorization problem, the theory  of multiple matrix orthogonal polynomials, non-intersecting Brownian motions and matrix Riemann-Hilber problem \cite{manasInverse2}-\cite{manas}. In fact the multicomponent 2D Toda hierarchy in \cite{manasinverse} is a periodic reduction of bi-infinite matrix-formed two dimensional Toda hierarchy. The coefficients(or dynamic variables) of the multicomponent 2D Toda hierarchy take values in complex finite-sized
matrix. The multicomponent 2D Toda hierarchy contains matrix-formed Toda equation as the first flow equation.

Adding additional logarithm flows to the Toda lattice hierarchy,
it becomes the extended Toda hierarchy\cite{CDZ} defined on a Lax operator
\begin{equation}
L=\Lambda+u+e^v\Lambda^{-1},\ \ u,v\in \C,
\end{equation}
 which governs the Gromov-Witten invariant of $CP^1$. The Gromov-Witten potential of $CP^1$ is actually a tau function of the extended
Toda hierarchy, i.e. the Gromov-Witten potential $\tau$ of $CP^1$
makes the  following Hirota quadratic equations\cite{M}  of the ETH
\begin{equation}\label{milanov}  \frac{d\lambda}{\lambda}
\(\Gamma^{\delta\#}\otimes \Gamma^{\delta}\) \( \Gamma^{\alpha}\otimes
\Gamma^{-\alpha}- \Gamma^{-\alpha}\otimes
\Gamma^{\alpha}\right)(\tau
\otimes \tau )
\end{equation}
regular in $\lambda$ computed at $q_{0}-q'_{0}=l\epsilon$
 for each  $l\in \Z$.
 The extended bigraded Toda
hierarchy(EBTH) is the extension of the bigraded Toda
hierarchy (BTH) which includes additional logarithm flows\cite{C,KodamaCMP}. This paper will tell us that the Hirota quadratic equations\eqref{milanov} can be derived as a reduction on Lie algebra from the Hirota bilinear equation of the extended $Z_N$-Toda hierarchy. Therefore the application of our Hirota bilinear equation of the extended $Z_N$-Toda hierarchy in Gromov-Witten theory becomes a great motivation of our study.  The Hirota bilinear equation of EBTH was equivalently constructed in our early paper\cite{ourJMP} and a very recent paper \cite{leurhirota}, because of the equivalence of $t_{1,N}$ flow and $t_{0,N}$ flow of EBTH. Meanwhile it was proved to govern Gromov-Witten invariant of the total
  descendent potential of $\mathbb{P}^1$ orbifolds \cite{leurhirota}.
A nature question is what about the corresponding extended multi-component  Toda hierarchy( as a matrix-formed generalization of extended Toda hierarchy\cite{CDZ}) and extended multicomponent bigraded Toda hierarchy.  There is a class of orbifolds which should be governed by some
logarithmic hierarchies. That is why we think this kind of new logarithmic hierarchy might be useful in
Gromov-Witten invariants theory governed by these two new hierarchies. With this motivation, our paper\cite{EMTH} will be devoted to construct a kind of Hirota quadratic equation taking values in a differential matrix
algebra set. This kind of Hirota bilinear equation might be useful in Gromov-Witten theory. In \cite{zuo}, a new hierarchy called as $Z_m$-KP hierarchy which take values in a maximal commutative subalgebra of $gl(m,\C)$ was constructed, meanwhile the relation between Frobenius manifold and dispersionless reduced $Z_m$-KP hierarchy was discussed. This inspires us to consider the Hirota quadratic equation of the commutative version of extended multi-component Toda hierarchy which might be useful in Frobenius manifold theory in this paper.

This paper is arranged as follows. In the next section we recall factorization problem and construct the logarithm matrix operators using which we define the extended flow of the multicomponent $Z_N$-Toda hierarchy. In Section 3,
we will give the Lax equations of the extended $Z_N$-Toda hierarchy (EZTH), meanwhile the multicomponent $Z_N$-Toda equations and the extended equations are introduced in this hierarchy. By Sato equations, Hirota bilinear equations of the EZTH are proved in
Section 4. The tau function of the EZTH will be defined in Section 5 which lead to the formalism of the generalized matrix Vertex operators and Hirota quadratic equations in Section 6. In section 7, multi-fold transformations of the EZTH will be constructed using determinant technique in \cite{rogueHMB}. To prove the integrability of this new hierarchy, Bi-hamiltonian structure and tau symmetry of the EZTH are constructed in Section 8. Section 9 will be devoted to a short conclusions and discussions.

\section{Factorization and Logarithm operators}

 Suppose that $\tilde G$ is a group which contains linear invertible elements of complex $N\times N$
complex matrices and its Lie algebra  $\tilde \g$  denotes the associative algebra  of complex $N\times N$
complex matrices $M_N(\C)$.
Now we will consider the linear space of function
$g:\R\rightarrow M_N(\C)$ with the shift operator $\Lambda$ acting on these functions as
$(\Lambda g)(x):=g(x+\epsilon)$. A Left
multiplication by  $X:\R\to \gl$ is as $X\Lambda^j$, $(
X\Lambda^j)(g)(x):=X(x)\circ g(x+j\epsilon)$ with defining the product
$(X(x)\Lambda^i)\circ(Y(x)\Lambda^j):=X(x)Y(x+i\epsilon)\Lambda^{i+j}.$
Then the set $\g$ of Laurent
series in $\Lambda$ as an associative algebra is a
Lie algebra under the standard commutator.

This Lie algebra has the following important splitting
\begin{gather}\label{splitting}
\g=\g_+\oplus\g_-,
\end{gather}
where
\begin{align*}
  \g_+&=\Big\{\sum_{j\geq 0}X_j(x)\Lambda^j,\quad X_j(x)\in\gl\Big\},&
  \g_-&=\Big\{\sum_{j< 0}X_j(x)\Lambda^j,\quad X_j(x)\in\gl\Big\}.
\end{align*}

The splitting
\eqref{splitting} leads us to consider the following factorization of
$g\in G$
\begin{gather}\label{fac1}
g=g_-^{-1}\circ g_+, \quad g_\pm\in G_\pm
\end{gather}
where $G_\pm$ have $\g_\pm$ as their Lie algebras. $G_+$
is the set of invertible linear operators  of the
form $\sum_{j\geq 0}g_j(x)\Lambda^j$; while $G_-$ is the set of
invertible linear operators of the form
$1+\sum_{j<0}g_j(x)\Lambda^j$. This algebra has a maximal commutative subalgebra $Z_N=\C[\Gamma]/(\Gamma^N)$ and $\Gamma=(\delta_{i,j+1})_{ij}\in gl(N,\C).$
Denote $Z_N(\Lambda):=\g_c$, then
we have the following splitting
\begin{gather}\label{splittingc}
\g_c=\g_{c+}\oplus\g_{c-},
\end{gather}
where
\begin{align*}
  \g_{c+}&=\Big\{\sum_{j\geq 0}X_j(x)\Lambda^j,\quad X_j(x)\in Z_N\Big\},&
  \g_{c-}&=\Big\{\sum_{j< 0}X_j(x)\Lambda^j,\quad X_j(x)\in Z_N\Big\}.
\end{align*}
The splitting \eqref{splittingc} leads us to consider the following factorization of
$g_c\in G_c$
\begin{gather}\label{fac1}
g_c=g_{c-}^{-1}\circ g_{c+}, \quad g_{c\pm}\in G_{c\pm},
\end{gather}
where $G_{c\pm}$ have $\g_{c\pm}$ as their Lie algebras. $G_{c+}$
is the set of invertible linear operators  of the
form $\sum_{j\geq 0}g_j(x)\Lambda^j$; while $G_{c-}$ is the set of
invertible linear operators of the form
$1+\sum_{j<0}g_j(x)\Lambda^j$.

 Now we
introduce  the following free operators $ W_0,\bar  W_0\in G_c$
\begin{align}
 \label{def:E}  W_0&:=\Exp{\sum_{j=0}^\infty
 t_{j}\frac{\Lambda^j}{\epsilon j!}+ s_{j}\frac{\Lambda^j}{\epsilon j!}(\epsilon \partial-c_j)},\ \ \partial=\frac{\partial}{\partial  x}, \\
\label{def:barE}   \bar W_0&:=\Exp{\sum_{j=0}^\infty
   t_{j}\frac{\Lambda^{-j}}{\epsilon j!}+ s_{j}\frac{\Lambda^{-j}}{\epsilon j!}(\epsilon \partial-c_j)}, \ \ c_j=\sum_{i=1}^j\frac 1i,
\end{align}
where $t_{j}, s_{j} \in \C$
will play the role of continuous times.

 We   define the dressing operators $W,\bar W$ as follows
\begin{align}
\label{def:baker}W&:=S\circ W_0,\ \  \bar W:=\bar S\circ \bar  W_0,\quad S\in G_{c-},\ \bar S\in G_{c+}.
\end{align}
%\end{definition}
Given an element $g\in G_c$ and denote $t=(t_{j}), s=(s_{j}); j\mathbb\in \N$, one can consider the factorization problem  in $G_c$ similarly as\cite{manasinverse}
\begin{gather}
  \label{facW}
  W\circ g=\bar W,
\end{gather}
i.e.
 the factorization problem
\begin{gather}
  \label{factorization}
  S(t,s)\circ W_0\circ g=\bar S(t,s)\circ\bar W_0.
\end{gather}
Observe that  $S,\bar S$ have expansions of the form
\begin{gather}
\label{expansion-S}
\begin{aligned}
S&=\I_N+\omega_1(x)\Lambda^{-1}+\omega_2(x)\Lambda^{-2}+\cdots\in G_{c-},\\
\bar S&=\bar\omega_0(x)+\bar\omega_1(x)\Lambda+\bar\omega_2(x)\Lambda^{2}+\cdots\in
G_{c+}.
\end{aligned}
\end{gather}
Also we define the symbols of $S,\bar S$ as  $\bS,\bar \bS$
\begin{gather}
\begin{aligned}
\bS&=\I_N+\omega_1(x)\lambda^{-1}+\omega_2(x)\lambda^{-2}+\cdots,\\
\bar \bS&=\bar\omega_0(x)+\bar\omega_1(x)\lambda+\bar\omega_2(x)\lambda^{2}+\cdots.
\end{aligned}
\end{gather}

The inverse operators $S^{-1},\bar S^{-1}$ of operators $S,\bar S$ have expansions of the form
\begin{gather}
\begin{aligned}
S^{-1}&=\I_N+\omega'_1(x)\Lambda^{-1}+\omega'_2(x)\Lambda^{-2}+\cdots\in G_{c-},\\
\bar S^{-1}&=\bar\omega'_0(x)+\bar\omega'_1(x)\Lambda+\bar\omega'_2(x)\Lambda^{2}+\cdots\in
G_{c+}.
\end{aligned}
\end{gather}
Also we define the symbols of $S^{-1},\bar S^{-1}$  as  $\bS^{-1},\bar \bS^{-1}$
\begin{gather}
\begin{aligned}
\bS^{-1}&=\I_N+\omega'_1(x)\lambda^{-1}+\omega'_2(x)\lambda^{-2}+\cdots,\\
\bar \bS^{-1}&=\bar\omega'_0(x)+\bar\omega'_1(x)\lambda+\bar\omega'_2(x)\lambda^{2}+\cdots.
\end{aligned}
\end{gather}

 The Lax  operators $\L\in G_c$
 are defined by
\begin{align}
\label{Lax}  \L&:=W\circ\Lambda\circ W^{-1}=\bar W\circ\Lambda^{-1}\circ \bar W^{-1},
\end{align}
%\end{definition}
and
%  \begin{enumerate}
%\item
%The Lax operators
have the following expansions
\begin{gather}\label{lax expansion}
\begin{aligned}
 \L&=\Lambda+u_1(x)+u_2(x)\Lambda^{-1}.
\end{aligned}
\end{gather}
 In fact the Lax  operators $\L\in G_c$
can also be equivalently defined by
\begin{align}
\label{Lax}  \L&:=S\circ\Lambda\circ S^{-1}=\bar S\circ\Lambda^{-1}\circ \bar S^{-1}.
\end{align}
These definitions are continuous interpolated version of the multi-component commutative Toda hierarchy, i.e. a continuous spatial parameter $x$ was brought into this hierarchy. Under this meaning, the continuous flow $\frac{\partial}{\partial x}$ is missing. To make these flows complete, we define the following logarithm matrix
\begin{align}
\log_+\L&=W\circ\epsilon \partial\circ W^{-1}=S\circ\epsilon \partial\circ S^{-1},\\
\log_-\L&=-\bar W\circ\epsilon \partial\circ \bar W^{-1}=-\bar S\circ\epsilon \partial\circ \bar S^{-1},
\end{align}
where $\d$ is the derivative about spatial variable $x$.

Combining these above logarithm operators together can derive following important logarithm matrix
\begin{align}
\label{Log} \log \L:&=\frac12(\log_+\L+\log_-\L)=\frac12(S\circ\epsilon \partial\circ S^{-1}-\bar S\circ\epsilon \partial\circ \bar S^{-1}):=\sum_{i=-\infty}^{+\infty}W_i\Lambda^i\in G_c,
\end{align}
which will generate a series of flow equations which contain the spatial flow in later defined Lax equations.

\section{ Lax equations of EZTH}

In this section we will use the factorization problem \eqref{facW} to derive  Lax equations.
Let us first introduce some convenient notations.
\begin{definition}Matrix operators $B_{j},D_{j}$ are defined as follows
\begin{align}\label{satoS}
\begin{aligned}
B_{j}&:=\frac{\L^{j+1}}{(j+1)!},\ \
D_{j}:=\frac{2\L^j}{j!}(\log \L-c_j),\ \  c_j=\sum_{i=1}^j\frac 1i,\ j\geq 0.
\end{aligned}
\end{align}
\end{definition}

Now we give the definition of the extended $Z_N$-Toda hierarchy(EZTH).
\begin{definition}The extended $Z_N$-Toda hierarchy is a hierarchy in which the dressing operators $S,\bar S$ satisfy the following Sato equations
\begin{align}
\label{satoSt} \epsilon\partial_{t_{j}}S&=-(B_{j})_-S,& \epsilon\partial_{t_{j}}\bar S&=(B_{j})_+\bar S,  \\
\label{satoSs}\epsilon\partial_{ s_{j}}S&=-(D_{j})_- S,& \epsilon\partial_{s_{j}}\bar S&=(D_{j})_+\bar S.\end{align}
\end{definition}
Then one can easily get the following proposition about $W,\bar W.$

\begin{proposition}The dressing operators $W,\bar W$ are subject to the following Sato equations
\begin{align}
\label{Wjk} \epsilon\partial_{t_{j}}W&=(B_{j})_+ W,& \epsilon\partial_{t_{j}}\bar W&=(B_{j})_+\bar W,  \\
\epsilon\partial_{s_{j}}W&=(\frac{\L^j}{j!}(\log_+ \L-c_j) -(D_{j})_-) W,& \epsilon\partial_{s_{j}}\bar W&=(-\frac{\L^j}{j!}(\log_- \L-c_j)+(D_{j})_+)\bar W.  \end{align}
\end{proposition}

 From the previous proposition we derive the following  Lax equations for the Lax operators.
\begin{proposition}\label{Lax}
 The  Lax equations of the EZTH are as follows
   \begin{align}
\label{laxtjk}
  \epsilon\partial_{t_{j}} \L&= [(B_{j})_+,\L],&
  \epsilon\partial_{s_{j}} \L&= [(D_{j})_+,\L],\
  \epsilon\partial_{ t_{j}} \log \L= [(B_{j})_+ ,\log \L],&
  \end{align}
   \begin{align}\epsilon(\log \L)_{ s_{j}}=[ -(D_{j})_-,\log_+ \L ]+
[(D_{j})_+ ,\log_- \L ].
\end{align}
\end{proposition}

To see this kind of hierarchy more clearly, the  $Z_N$-Toda equations as the $t_{0}$ flow equations  will be given in next subsection.
\subsection{The extended $Z_N$-Toda equations}
 As a consequence of the factorization problem \eqref{facW} and  Sato equations, after taking into account that   $S\in G_{c-}$ and $\bar S\in G_{c+}$, the $t_0$ flow of $\L$ in the form of $\L=\Lambda+U+V\Lambda^{-1}$ is as
\begin{gather}\label{exp-omega}
\begin{aligned}
  \epsilon\partial_{t_{0}} \L&= [\Lambda+U,V\Lambda^{-1}],
  \end{aligned}
\end{gather}
which lead to matrix-Toda equation
\[\epsilon\partial_{t_{0}} U&=& V(x+\epsilon)-V(x),\\ \label{toda}
\epsilon\partial_{t_{0}} V&=& U(x)V(x)-V(x)U(x-\epsilon).\]
Of course, one can switch the order of the matrices because of the commutativity of $Z_N$.
Suppose
\[U=\begin{bmatrix}u_0&0\\ u_1&u_0\end{bmatrix},\ \ V=\begin{bmatrix}v_0&0\\ v_1&v_0\end{bmatrix},\]
then the specific coupled Toda equation is
\[\epsilon\partial_{t_{0}} u_0&=& v_0(x+\epsilon)-v_0(x), \\
\epsilon\partial_{t_{0}} u_1&=& v_1(x+\epsilon)-v_1(x),\\ \label{stoda}
\epsilon\partial_{t_{0}} v_0&=& u_0(x)v_0(x)-v_0(x)v_0(x-\epsilon),\\
\epsilon\partial_{t_{0}} v_1&=&(u_1(x)-u_1(x-\epsilon))v_0(x)-v_1(x)(u_0(x)-u_0(x-\epsilon).\]
 To get the standard matix-Toda equation, one need to use the alternative expressions
\begin{gather}\label{exp-omega1}
\begin{aligned}
  U&:=\omega_1(x)-\omega_1(x+\epsilon)=\epsilon\partial_{t_1}\phi(x),\\
 V&= \Exp{\phi(x)}\Exp{-\phi(x-\epsilon)}=-\epsilon\partial_{t_1}\omega_1(x).
\end{aligned}
\end{gather}

From Sato equation we deduce the following set of nonlinear
partial differential-difference equations
\begin{align}\left\{
\begin{aligned}
 \omega_1(x)-\omega_1(x+\epsilon)&=\epsilon\partial_{t_1}(\Exp{\phi(x)})\cdot\Exp{-\phi(x)},\\
\epsilon\partial_{t_1}\omega_1(x)&=-\Exp{\phi(x)}\Exp{-\phi(x-\epsilon)}.\end{aligned}\right.
\label{eq:multitoda}
\end{align}
Observe that if we cross the two first equations, then we get
\begin{align*}
  \epsilon^2\partial_{t_1}^2\phi(x)=
  \Exp{\phi(x+\epsilon)}\Exp{-\phi(x)}-\Exp{\phi(x)}\Exp{-\phi(x-\epsilon)}
\end{align*}
which is the $N\times N$ matrix-valued extension of the Toda equation, from which the original Toda equation appears for $N=1$.
When $N=2$, the equation for $\phi=\begin{bmatrix}\phi_0&0\\ \phi_1&\phi_0\end{bmatrix}$ is the following coupled Toda system
\begin{align*}
  \epsilon^2\partial_{t_1}^2\phi_0(x)&=
  \Exp{\phi_0(x+\epsilon)-\phi_0(x)}-\Exp{\phi_0(x)-\phi_0(x-\epsilon)},\\
  \epsilon^2\partial_{t_1}^2\phi_1(x)&=
  (\phi_1(x+\epsilon)-\phi_1(x))\Exp{\phi_0(x+\epsilon)-\phi_0(x)}-  (\phi_1(x)-\phi_1(x-\epsilon))\Exp{\phi_0(x)-\phi_0(x-\epsilon)}.\\
\end{align*}
In the calculation, the identity $e^{\begin{bmatrix}\phi_0&0\\ \phi_1&\phi_0\end{bmatrix}}=\begin{bmatrix}e^{\phi_0}&0\\ \phi_1e^{\phi_0}&e^{\phi_0}\end{bmatrix}$
is used.
Besides above $Z_N$-Toda equations, with logarithm flows the EZTH also contains some extended flow equations in the next part.
Here we consider  the extended flow equations in the simplest case, i.e. the $s_{0}$ flow for $\L=\Lambda+u_0+u_1\Lambda^{-1},$
\[\epsilon \d_{s_{0}}\L&=&[(S\epsilon \d_x S^{-1})_+,\L]\\
&=&[\epsilon \d_xS S^{-1},\L]\\
&=&\epsilon\L_x,\]
which leads to  the following specific equation
\[\d_{s_{0}}U&=& U_{x} ,\ \ \d_{s_{0}}V= V_{x} .\]

To see the extended equations clearly, one need to rewrite the extended flows  in the Lax equations of the EZTH as in the following lemma.

\begin{lemma}\label{modifiedLax}

The extended flows in Lax formulation of the EZTH can be equivalently given
by
\begin{equation}
  \label{edef2}
\epsilon\frac{\partial \L}{\partial s_{j}} = [D_{j} ,\L ],
\end{equation}
\begin{align}
  &D_{j} = (\frac{\L^j}{j!}(\log_+ \L-c_j))_+-(\frac{\L^j}{j!}(\log_- \L-c_j))_-,
\end{align}
which can also be rewritten in the form
\begin{equation}
  \label{edef2'}
\epsilon \frac{\partial \L}{\partial s_n} = [\bar  D_n ,\L ],
\end{equation}
\begin{align} \notag
 \bar D_{j} &= \frac{\L^j}{j!}\epsilon \d+ [ \frac{\L^j}{j!} (\sum_{k < 0} W_k(x) \Lambda^k -c_j)]_+-[ \frac{\L^j}{j!}  (\sum_{k \geq 0} W_k(x) \Lambda^k -c_j)]_-.
\end{align}

\end{lemma}

Then one can derive the $s_1$ flow equation of the EZTH as
\[\notag &\epsilon U_{s_1}=(1-\Lambda)(V(\Lambda^{-1}-1)^{-1}\epsilon (\log V)_x)-2(\Lambda-1)V+\frac{\epsilon}2U_x^2+\epsilon V_x,\\
\notag &\epsilon V_{s_1}=((\Lambda^{-1}-1)^{-1}\epsilon V_xV^{-1}+2)(U(x-\epsilon)-U(x))V
+\epsilon V_xU(x-\epsilon)+\epsilon (U_x(x-\epsilon)+ U_x(x))V,\]
where $U,V$ without bracket behind them means $U(x),V(x)$ respectively.
To give a linear description of the EZTH, we introduce matrix wave functions  $\psi,\bar\psi$ in the following part.

The matrix wave functions of the EZTH are
defined by
\begin{gather}\label{baker-fac}
\begin{aligned}
\psi&= W\cdot\chi, &
\bar\psi&=\bar W\cdot \bar\chi,
\end{aligned}
\end{gather}
where
\[
\chi(z):=z^{\frac{x}{\epsilon}}\mathbb I_N,\ \ \bar \chi(z):=z^{-\frac{x}{\epsilon}}\mathbb I_N,\
\]
and the $``\cdot"$ means the action of an operator on a function.
Note that $\Lambda\cdot\chi=z\chi$ and  the following asymptotic expansions
can be defined
\begin{gather}\label{baker-asymp}
\begin{aligned}
  \psi&=z^{\frac{x}{\epsilon}}(\I_N+\omega_1(x)z^{-1}+\cdots)\,\psi_0(z),&\psi_0&:=
 \Exp{\sum_{j=1}^\infty t_{j}\frac{z^j}{\epsilon j!}+ s_{j}\frac{z^j}{\epsilon j!}( \log z-c_j)},& z&\rightarrow\infty,\\
\bar\psi&=z^{-\frac{x}{\epsilon}}(\bar\omega_0(x)+\bar\omega_1(x)z+\cdots)\,\bar\psi_0(z),
&\bar\psi_0&:=
\Exp{\sum_{j=0}^\infty
   t_{j}\frac{z^{-j}}{\epsilon j!}+ s_{j}\frac{z^{-j}}{\epsilon j!}(\log z-c_j)},& z&\rightarrow 0.
\end{aligned}
\end{gather}

We can further get linear equations in the following proposition.

\begin{proposition}The matrix wave functions $\psi,\bar\psi$ are subject to the following Sato equations
\begin{align}
 \L\cdot\psi&=z\psi,\ \ \ &&\L\cdot\bar\psi=z\bar\psi,\\
 \epsilon\partial_{t_j}\psi&=(B_{j})_+\cdot \psi,& \epsilon\partial_{t_j}\bar \psi&=(B_{j})_+\cdot\bar \psi,  \\
\epsilon\partial_{s_{j}}\psi&=(\frac{\L^j}{\epsilon j!}(\log_+ \L-c_j) -(D_{j})_-)\cdot \psi,& \epsilon\partial_{s_{j}}\bar \psi&=(-\frac{\L^j}{\epsilon j!}(\log_- \L-c_j)+(D_{j})_+)\cdot\bar \psi.  \end{align}
\end{proposition}

\section{Hirota bilinear equations}
From Lax equations, one can find the  $s_{0}$ flow is equivalent to the spatial flow $\d_x$.
Basing on this fact, Hirota bilinear equations which are equivalent to the Lax equations of the EZTH can be derived in the following proposition.
\begin{proposition}\label{HBEoper}
 $W$ and $\bar W$ are matrix-valued wave operators of the extended $Z_N$-Toda hierarchy if and only the following Hirota bilinear equations hold
\begin{align}
W\Lambda^r W^{-1}&=\bar W\Lambda^{-r}\bar W^{-1}, \ r\in \N.
   \end{align}
\end{proposition}

\begin{proof}
$\Rightarrow$ Set
\begin{align}
\ga&=(\ga_{0},\ga_{1},\ga_{2},\ldots;),\ \
 \beta=(\beta_{1},\beta_{2},\ldots), \end{align}
be a multi index and
\begin{align}
\d^\ga:&=\d_{t_{0}}^{\ga_{0}}\d_{t_{0}}^{\ga_{1}}\d_{t_{2}}^{\ga_{2}}\ldots\ ,\  \
\d^\beta:=\d_{s_{1}}^{\beta_{1}}\d_{s_{2}}^{\beta_{2}}\ldots\ .
\end{align}
Suppose $\d^{\theta}=\d^{\alpha}\d^{\beta}$ . Firstly we shall prove the left
statement leads to
\begin{eqnarray} \label{HBE2} W (x,t,\Lambda)\Lambda^r
W^{-1}(x,t',\Lambda) = \bar W (x,t,\Lambda)
\Lambda^{-r}\bar W^{-1}(x,t',\Lambda)
\end{eqnarray}
 for all integers $r\geq 0$.
Using the same as the method used in\cite{M,ourJMP}, by induction on $\ga $,
we shall prove that
\begin{equation} \label{2.7}
W(x,t,\Lambda)\Lambda^r(\d^\theta W^{-1}(x,t,\Lambda))
=\bar W(x,t,\Lambda) \Lambda^{-r}(\d^\theta\bar W^{-1}(x,t,\Lambda)).
\end{equation} When $\theta=0$, it is obviously true according to the definition of
matrix-valued wave operators.\\
Suppose eq.\eqref{2.7} is true in the case of $\theta\neq 0$.
 Note that

\begin{equation}
\notag\epsilon\partial_{p_{j}}
W :=
\begin{cases}
 [(\d_{t_{j}}S)S^{-1}+S  \Lambda^jS^{-1}]W,
&p_{j}=t_{j},\\
[(\d_{s_{j}}S)S^{-1}+S  \Lambda^j\d_xS^{-1}]W,  &p_{j}=s_{j},
  \end{cases}
\end{equation}
and
\begin{equation}
\notag\epsilon\partial_{p_{j}}
\bar W :=
\begin{cases}
 (\d_{t_{j}}\bar S)\bar S^{-1}\bar W,
&p_{j}=t_{j},\\
[(\d_{\bar S_{j}}\bar S)\bar S^{-1}+\bar S  \Lambda^{-j}\d_x\bar S^{-1}]\bar W,  &p_{j}=s_{j},
  \end{cases}
\end{equation}
which further lead to \\
\begin{equation}
\notag\epsilon\partial_{p_{j}}
W :=
\begin{cases}
 (B_{j})_+W,
&p_{j}=t_{j},\\
[-(D_{j})_-+  \frac{\L^j}{\epsilon j!}(\log_+ \L-c_j) ]W,  &p_{j}=s_{j},
  \end{cases}
\end{equation}
and
\begin{equation}
\notag\epsilon\partial_{p_{j}}
\bar W :=
\begin{cases}
 (B_{j})_+\bar W,
&p_{j}=t_{j},\\
[(D_{j})_+-  \frac{\L^j}{\epsilon j!}(\log_- \L-c_j) ]\bar W,  &p_{j}=s_{j}.
  \end{cases}
\end{equation}

This further implies
\begin{equation} \notag (\d_{p_{j}}W)\Lambda^{r}(\d^\theta W^{-1}) =
(\d_{p_{j}}\bar W)\Lambda^{-r}(\d^\theta\bar W^{-1}) \end{equation}
by considering \eqref{2.7} and furthermore we  get  \begin{equation} \notag
W\Lambda^r(\d_{p_{j}}\d^\theta W^{-1} )=
\bar W\Lambda^{-r}(\d_{p_{j}}\d^\theta\bar W^{-1}).
\end{equation} Thus if we increase the power of  $\d_{p_{j}}$ by 1,
eq.\eqref{2.7} still holds.
 The induction is completed.
Taylor expanding both sides of eq.\eqref{HBE2} about $t=t',s=s'$, one can finish
the proof of eq.\eqref{HBE2}.

$\Leftarrow$ Vice versa, by separating the negative and the positive part
of the equation, we can prove
$S, \ \bar S$ are a pair of matrix-valued wave operators.
\end{proof}

To give a description in terms of matrix-valued wave functions, the following symbolic definitions are needed.

If the series have forms
\begin{eqnarray*} W(x,t,s,\Lambda)=\sum_{i\in \Z} a_i(x,t,s,
\d_x)\Lambda^i \mbox{ and } \bar W(x,t,s,\Lambda)=\sum_{i\in \Z}
b_i(x,t,s, \d_x)\Lambda^{i}, \end{eqnarray*}

\begin{eqnarray*}  W^{-1}(x,t,s,\Lambda)=\sum_{i\in \Z}\Lambda^{i} a_i'(x,t,s,
\d_x) \mbox{ and } \bar W^{-1}(x,t,s,\Lambda)=\sum_{j\in
\Z}\Lambda^{j}b_j'(x,t,s, \d_x),  \end{eqnarray*} then we denote their corresponding
left symbols $\W$,  $\bar \W$ and right symbols $\W^{-1}$, $\bar \W^{-1}$
as follows
\begin{eqnarray*}
&&  \W(x,t,s,\lambda) =\sum_{i\in \Z} a_i(x,t,s,
\d_x)\lambda^i,\ \  \W^{-1}(x,t,s,\lambda)=  \sum_{i\in \Z} a_i'(x,t,s,
\d_x)\lambda^{i},\\
%%%%%%%%%%%%%%%%%%%%%%%%%%%%%%%%%%%%%%%%%%%%%
&&
\bar \W(x,t,s,\lambda) =\sum_{i\in \Z} b_i(x,t,s,
\d_x)\lambda^{i},\ \ \bar \W^{-1}(x,t,s,\bar t,\lambda)=\sum_{j\in
\Z}b_j'(x,t,s, \d_x)\lambda^{j}.
\end{eqnarray*}
With above preparation, it is time to give another form of Hirota bilinear equation(see the following proposition) after defining residue as $\res_{\lambda }\sum_{n\in \Z}\alpha_n \lambda^n=\alpha_{-1}$ using the similar proof as \cite{UT,M,ourJMP}.
\begin{proposition}\label{wave-operators}
Let  $s_{0} =s'_{0},$
 $S$ and
$\bar S$ are matrix-valued wave operators of the $Z_N$-Toda hierarchy if and only if for all   $m\in
\Z$, $r\in \N$ , the following Hirota bilinear identity hold

\begin{eqnarray}  \notag &&\res_{\lambda }
 \left\{
\lambda^{r+m-1}\ \W(x,t,s,\epsilon \partial_x,\lambda) \W^{-1}(x-m\epsilon,t',s', \epsilon \partial_x,\lambda)
\right\} = \\ \label{HBE3}&& \res_{\lambda }
 \left\{
\lambda^{-r+m-1}\bar \W( x,t,s,\epsilon \partial_x,\lambda )\
\bar \W^{-1}(x-m\epsilon,t',s',\epsilon \partial_x,\lambda) \right\}.
\end{eqnarray}
\end{proposition}

\begin{proof}
 Let $m\in \Z$, $r\in \N$ and $s_{0} = s'_{0}$.
Put
\begin{eqnarray*} W(x,t,s,\Lambda)=\sum_{i\in \Z} a_i(x,t,s,
\d_x)\Lambda^i \mbox{ and } \bar W(x,t,s,\Lambda)=\sum_{i\in \Z}
b_i(x,t,s, \d_x)\Lambda^{i}, \end{eqnarray*}

\begin{eqnarray*} W^{-1}(x,t,s,\Lambda)=\sum_{i\in \Z}\Lambda^{i} a_i'(x,t,s,
\d_x) \mbox{ and }  \bar W^{-1}(x,t,\Lambda)=\sum_{j\in
\Z}\Lambda^{j}b_j'(x,t,s, \d_x) , \end{eqnarray*} and compare the
coefficients in front of $\Lambda^{-m}$ in eq.\eqref{HBE2}:
\begin{eqnarray*} \sum_{i+j=-m-r} a_i(x,t,s, \d_x)a_j'(x-m\epsilon,t',s',
\d_x) = \sum_{i+j=-m+r}  b_i(x,t,s,\d_x)b_j'(x-m\epsilon,t',s', \d_x).
\end{eqnarray*} This equality can be written also as eq.\eqref{HBE3}.

\end{proof}

To give Hirota quadratic function in terms of tau functions, we need to define and prove the existence of tau function of the EZTH firstly in the next section.

\section{Tau-functions of EZTH}

Introduce the following sequences:
\[t-[\lambda] &:=& (t_{j}-
  \epsilon(j-1)!\lambda^j, 0\leq j\leq \infty).
\]
A matrix-valued function $\tau\in Z_N$  depending only on the dynamical variables $t$ and
$\epsilon$ is called the  {\em \bf Matrix tau-function of the EZTH} if it
provides symbols related to matrix-valued wave operators as following,

\begin{eqnarray}\label{Mpltaukk}\bS: &=&\frac{ \tau
(s_{0}+x-\frac{\epsilon}{2}, t_{j}-\frac{\epsilon(j-1)!}{\lambda^j},s;\epsilon) }
     {\tau (s_{0}+x-\frac{\epsilon}{2},t,s;\epsilon)},\\
     \label{Mpl-1taukk}\bS^{-1}: &=&\frac{ \tau
(s_{0}+x+\frac{\epsilon}{2}, t_{j}+\frac{\epsilon(j-1)!}{\lambda^j},s;\epsilon) }
     {\tau (s_{0}+x+\frac{\epsilon}{2},t,s;\epsilon)},\\ \label{Mprtaukk}
\bar \bS:&= &\frac{ \tau
(s_{0}+x+\frac{\epsilon}{2},t_{j}+\epsilon(j-1)!\lambda^j,s;\epsilon)}
     {\tau(s_{0}+x-\frac{\epsilon}{2},t,s;\epsilon)},\\
     \bar \bS^{-1}:&= &\frac{\tau
(s_{0}+x-\frac{\epsilon}{2},t_{j}-\epsilon(j-1)!\lambda^j,s;\epsilon)}
     {\tau(s_{0}+x+\frac{\epsilon}{2},t,s;\epsilon)}.
     \end{eqnarray}
     Here the division means the multiplication of the numerator matrix by the inverse of the denominator matrix.
One can get the solution $U,V$ in terms of tau functions
as
\[U=(\log \tau)_{xx},V=\log\frac{\tau(x+\epsilon)\tau(x-\epsilon)}{\tau^2(x)}\]

When $N=2$
\[\begin{bmatrix}u_0&0\\u_1&u_0\end{bmatrix}&=&\begin{bmatrix}(\log\tau_0)_{xx}&0\\ (\frac{\tau_1}{\tau_0})_{xx}&(\log\tau_0)_{xx}\end{bmatrix},\\
e^{\begin{bmatrix}v_0&0\\v_1&v_0\end{bmatrix}}&=&\begin{bmatrix}e^{v_0}&0\\v_1e^{v_0}&e^{v_0}\end{bmatrix}\\
&=&\begin{bmatrix}\frac{\tau_0(x+\epsilon)\tau_0(x-\epsilon)}{\tau_0}&0\\ \frac{\tau_1(x+\epsilon)\tau_0(x-\epsilon)+\tau_0(x+\epsilon)\tau_1(x-\epsilon)}{\tau_0}-
\frac{\tau_0(x+\epsilon)\tau_0(x-\epsilon)\tau_1}{\tau_0^2}&\frac{\tau_0(x+\epsilon)\tau_0(x-\epsilon)}{\tau_0}\end{bmatrix}.\]
This implies
\[u_0&=&(\log\tau_0)_{xx},\ \ u_1=(\frac{\tau_1}{\tau_0})_{xx},\\
v_0&=&\log\frac{\tau_0(x+\epsilon)\tau_0(x-\epsilon)}{\tau_0},\ \ v_1=(\Lambda-1+\Lambda^{-1})\frac{\tau_1(x)}{\tau_0(x)}.\]

By the  Proposition \ref{wave-operators}, one can prove the following lemma
\begin{lemma}The following equations hold
\label{identities2s'}
\begin{eqnarray} \label{lr2}
&&\sum_{k=1}^{N}\bS(x,t,\lambda_{1})_{ik}\bS^{-1}(x+\epsilon,t+[\lambda_{2}],\lambda_{1})_{kj}= \sum_{k=1}^{N}\bar \bS(x,t,\lambda_{2})_{ik}\bar \bS^{-1}(x,t-[\lambda_{1}^{-1}],\lambda_{2})_{kj},\\
\label{i22}&&
\sum_{k=1}^{N}\bS(x,t,\lambda_{1})_{ik}\bS^{-1}(x,t-[\lambda_{2}^{-1}],\lambda_{1})_{kj}
=\sum_{k=1}^{N}\bS(x,t,\lambda_{2})_{ik}\bS^{-1}(x,t-[\lambda_{1}^{-1}],\lambda_{2})_{kj},\\
\notag&& \notag
 \sum_{k=1}^{N}\bar \bS(x,t,\lambda_{1})_{ik}\bar \bS^{-1}(x+\epsilon,t+[\lambda_{2}],\lambda_{1})_{kj}
=\sum_{k=1}^{N}\bar \bS(x,t,\lambda_{2})_{ik}\bar \bS^{-1}(x+\epsilon,t+[\lambda_{1}],\lambda_{2})_{kj}.
\end{eqnarray}
\end{lemma}

Using Lemma \ref{identities2s'}, we can prove the following important proposition which gives the existence of matrix-valued tau functions.

\begin{proposition}\label{tau-function}
Given a pair of wave operators $\bS$ and $\ \bar \bS$ of the EZTH there
exists corresponding matrix-valued invertible tau-functions $\tau\in Z_N$, which is unique up to the
multiplication by a non-vanishing function independent of $t_{j}, j\geq 1$.
\end{proposition}
\begin{proof} Here, we shall note that the $Z_N$-valued tau function $\tau(x,\t)$
corresponding to the wave operators $\bS$ and  $\ \bar \bS$ is in fact
$\tau(x-\epsilon/2,\t)$.
\\The system is equivalent to:
\begin{eqnarray*}
&& \label{eq1} \log \bS = \(\exp\left({-\epsilon\sum_{j=0}^\infty j!\lambda^{-(j+1)}\partial_{t_{j}}}\right)-1\)\log \tau, \\
\label{eq2} && \log  \bar \bS = \(\exp\left(\epsilon\d_{x}+\epsilon\sum_{j=0}^\infty j!\lambda^{j+1}\partial_{t_{j}}\right)-1\)\log \tau, \\
\label{eq3} &&\d_{s_{0}} \log \tau(x,\t) = \d_x \log \tau(x,\t).
\end{eqnarray*}
Then using Lemma \ref{identities2s'} will help us to derive the existence of the tau function of this hierarchy.

 \end{proof}

After giving tau functions of the EZTH, what is the Hirota bilinear equation in form of tau function becomes a natural question which will be answered in the next
section with the help of generalized Vertex operators.

\section{Generalized matrix Vertex operators and  Hirota quadratic equations}
In this section we continue to  discuss on the fundamental properties
of the tau function of the EZTH, i.e., the Hirota quadratic equations of the EZTH. So we
introduce the following vertex operators
\begin{eqnarray*}
\Gamma^{\pm a} :&=&\exp\left(\pm \frac{1}{\epsilon}
(\sum_{j=0}^\infty t_{j}\frac{\lambda^{j+1}}{(j+1)!}+ s_{j}\frac{\lambda^{j}}{j!}( \log \lambda-c_j))\right)\times\exp\left({\mp
\frac{\epsilon}{2}\partial_{s_{0}} \mp [\lambda^{-1}]_\d  }\right),\\
\Gamma^{\pm b} :&=&\exp\left(\pm  \frac{1}{\epsilon}
(\sum_{j=0}^\infty t_{j}\frac{\lambda^{-j-1}}{(j+1)!}- s_{j}\frac{\lambda^{-j}}{j!}( \log \lambda-c_j))\right)\times\exp\left({\mp
\frac{\epsilon}{2}\partial_{s_{0}} \mp [\lambda]_{\d}  }\right),
\end{eqnarray*}

where\ \
\begin{eqnarray*}
 [\lambda]_\d  :&=&
\epsilon\sum_{j=0}^\infty j!\lambda^{j+1}\partial_{t_{j}}.
\end{eqnarray*}

Because of the logarithm $\log \lambda$, the vertex
operators  $\Gamma^{\pm a} \otimes \Gamma^{\mp a}$ and
$\Gamma^{\pm b} \otimes \Gamma^{\mp b }$  are multi-valued
function. There
are monodromy factors $M^a$ and $M^b$ respectively as following
among different branches around $\lambda=\infty$
\begin{equation} M^{a}= \exp \left\{ \pm \frac{2\pi i}{\epsilon}
\sum_{j\geq 0}\frac{\lambda^{j}}{j!} ( s_{j} \otimes 1 - 1\otimes s_{j})
\right\},\end{equation}
\begin{equation}
 M^{b}= \exp \left\{ \pm \frac{2\pi i}{\epsilon}
\sum_{j\geq 0}\frac{\lambda^{-j}}{j!} ( s_{j} \otimes 1 - 1\otimes s_{j})
\right\}.
\end{equation}
In order to offset the complication, we need to generalize the
concept of Vertex operators which leads it to be not scalar-valued
any more but take values in a differential operator algebra in $Z_N$. So we introduce the following vertex operators
\begin{equation}\Gamma_{a} = \exp\( -\sum_{j>0}\frac{j!\lambda^{j+1}}{\epsilon
}(\epsilon\d_x)s_{j}\) \exp(x\partial_{s_{0}}),\end{equation}
\begin{equation}\Gamma_{b} = \exp\( -\sum_{j>0}\frac{j!\lambda^{-(j+1)}}{\epsilon
}(\epsilon\d_x)s_{j}\) \exp(x\partial_{s_{0}}),\end{equation}
\begin{equation}\Gamma^{\#}_{a} =\exp(x\partial_{s_{0}}) \exp\( \sum_{j>0}\frac{j!\lambda^{j+1}}{\epsilon
}(\epsilon\d_x)s_{j}\) ,\end{equation}
\begin{equation}\Gamma^{\#}_{b} =  \exp(x\partial_{s_{0}})\exp\( \sum_{j>0}\frac{j!\lambda^{-(j+1)}}{\epsilon
}(\epsilon\d_x)s_{j}\).\end{equation}
 Then \begin{equation}
 \label{double delta a} \Gamma^{\#}_{a}\otimes \Gamma_{a} = \exp(x\partial_{s_{0}})\exp\(
\sum_{j>0}\frac{j!\lambda^{j+1}}{\epsilon
}(\epsilon\d_x)(s_{j}-s'_{j}) \) \exp(x\partial_{s'_{0}}),
\end{equation}
\begin{equation}
\label{double delta b} \ \ \ \Gamma^{\#}_{b}\otimes
\Gamma_{b} = \exp(x\partial_{s_{0}})\exp\(
\sum_{j>0}\frac{j!\lambda^{-(j+1)}}{\epsilon
}(\epsilon\d_x)(s_{j}-s'_{j}) \) \exp(x\partial_{s'_{0}}).
\end{equation}

After some computation we get
\begin{eqnarray*} && \(\Gamma^{\#}_{a}\otimes \Gamma_{a} \) M^{a} =
\exp \left\{ \pm \frac{2\pi i}{\epsilon} \sum_{j> 0} \frac{\lambda^{j}}{j!} ( s_{j}-s'_{j})
\right\}\\
&& \exp\(  \pm \frac{2\pi i}{\epsilon} ((s_{0}+x) -(s'_{0}+x+ \sum_{j> 0} \frac{\lambda^{j}}{j!} ( s_{j}-s'_{j})) \) \(\Gamma^{\#}_{a}\otimes \Gamma _{a}\)
\\&=& \exp\({\pm \frac{2\pi i}{\epsilon}(s_{0}-s'_{0})}\)
\(\Gamma^{\#}_{a}\otimes \Gamma _{a}\),
\end{eqnarray*}
\begin{eqnarray*}
&& \(\Gamma^{\#}_{b}\otimes \Gamma_{ b} \) M^{b} =
\exp \left\{ \pm \frac{2\pi i}{\epsilon} \sum_{j> 0} \frac{\lambda^{-j}}{j!} ( s_{j}-s'_{j})
\right\}\\
&& \exp\(  \pm \frac{2\pi i}{\epsilon} ( (s_{0}+x) -(
s'_{0}+x+ \sum_{j> 0} \frac{\lambda^{-j}}{j!} ( s_{j}-s'_{j})) \)\(\Gamma^{\#}_{b}\otimes \Gamma_{b} \)
\\&=& \exp\({\pm \frac{2\pi i}{\epsilon}(s_{0}-s'_{0})}\)
\(\Gamma^{\#}_{b}\otimes \Gamma_{b} \).
\end{eqnarray*}
Thus when $s_{0}-s'_{0} \in \Z\epsilon $, $\(\Gamma^{\#}_{a}\otimes \Gamma_{a}\) \( \Gamma^{a}\otimes
\Gamma^{-a}\) \mbox{and}\(\Gamma^{\#}_{b}\otimes
\Gamma_{b }\)\(\Gamma^{-b}\otimes\Gamma^{b}\)$ are all
single-valued near $\lambda=\infty$.

 Now we should note that the above vertex operators  take
 value in a $Z_N$-valued differential operator algebra $\C[\d,x,t,s,\epsilon]:=\{f(x,t,\epsilon)|f(x,t,s,\epsilon)=\sum_{i\geq 0}\sum_{k\geq 0}^Nc_{ik}(x,t,s,\epsilon)\Gamma^k\d^i\}$.

\begin{theorem}\label{t11}
The invertible $Z_N$-valued matrix $\tau(t,s,\epsilon)$  is a tau-function of the EZTH if and only if it
satisfies the  following Hirota quadratic equations  of the EZTH,
\begin{equation} \label{HBE} \res_{{\rm{\lambda}}}
 \lambda^{r-1}\(\Gamma^{\#}_{a}\otimes \Gamma_{a}\) \( \Gamma^{a}\otimes
\Gamma^{-a}\right)(\tau
\otimes \tau ) =\res_{{\rm{\lambda}}}\lambda^{-r-1}\(\Gamma^{\#}_{b}\otimes
\Gamma_{b }\)\(\Gamma^{-b}\otimes\Gamma^{b} \) (\tau
\otimes \tau )
\end{equation}
computed at $s_{0}-s'_{0}=l\epsilon$
 for each  $l\in \Z$, $r\in \N$.
\end{theorem}
 \begin{proof}

 We just need  to prove that the HBEs
are equivalent to the right side in Proposition
\ref{wave-operators}. By a straightforward computation we can get
the following four identities {\allowdisplaybreaks}
\begin{eqnarray}\label{vertex computation1}
\Gamma^{\#}_a \Gamma^{a}\tau & =& \tau(s_{0}+x-\epsilon/2,t,s)
\lambda^{\ s_{0}/\epsilon} \W(x,t,s,\epsilon \d_x,\lambda )\lambda^{\I_Nx/\epsilon},
\\ \label{vertex computation2}
 \Gamma_a \Gamma^{-a}\tau  & =&
\lambda ^{-(s_{0}+x)/\epsilon}
\W^{-1}(x,t,s,\epsilon\d_x,\lambda )\tau(x+s_{0}+\epsilon/2,t,s), \\\label{vertex
computation3}
 \Gamma^{\#}_b \Gamma^{-b}\tau  & =&
\tau(x+s_{0}-\epsilon/2,t,s) \lambda^{s_{0}/\epsilon} \bar \W(x,t,s,\epsilon \d_x,\lambda
)\lambda^{ x\I_N/\epsilon}, \\\label{vertex computation4}
 \Gamma_b \Gamma^{b}\bar\tau & = &\lambda^{-s_{0}/\epsilon} \lambda^{
 -x\I_N/\epsilon} \bar \W^{-1}(x,t,s,\epsilon \d_x,\lambda)\
\tau(x+s_{0}+\epsilon/2,t,s) .
\end{eqnarray}
The proof of four equations eq.\eqref{vertex
computation1}-eq.\eqref{vertex computation4} can be derived by a similar method as in \cite{M,ourJMP}.
By substituting four equations eq.\eqref{vertex
computation1}-eq.\eqref{vertex computation4} into the HBEs
\eqref{HBE},
eq.\eqref{HBE3} is derived.
\end{proof}

 Doing a transformation on the eq.\eqref{HBE} by $\lambda\rightarrow \lambda^{-1},$ then the eq.\eqref{HBE} becomes
\begin{equation} \label{HBE'} \res_{{\rm{\lambda}}}
 \lambda^{r-1}\left(\(\Gamma^{\#}_{a}\otimes \Gamma_{a}\) \( \Gamma^{a}\otimes
\Gamma^{-a}-\Gamma^{-a}\otimes\Gamma^{a} \) \right)(\tau
\otimes \tau )=0
\end{equation}
computed at $s_{0}-s'_{0}=l\epsilon$
 for each  $l\in \Z$, $r\in \N$.
 That means
 \begin{equation} \label{HBEmilanov}\frac{d \lambda}{\lambda} \left(\(\Gamma^{\#}_{a}\otimes \Gamma_{a}\) \( \Gamma^{a}\otimes
\Gamma^{-a}-\Gamma^{-a}\otimes\Gamma^{a} \) \right)(\tau
\otimes \tau )
\end{equation}
is regular in $\lambda$
computed at $s_{0}-s'_{0}=l\epsilon$
 for each  $l\in \Z$. The eq.\eqref{HBEmilanov} in the case when $N=1$ is exactly the
  Hirota quadratic equation \eqref{milanov} of the extended Toda hierarchy in \cite{M}.
As we know, the Vertex operator in fact gives one special Backlund transformation of the EZTH. To give more information on the relations among different
solutions of the EZTH, the Darboux transformation of the EZTH will be constructed using kernel determinant technique as \cite{Hedeterminant,rogueHMB} in the next section.

\section{Darboux transformations of the EZTH}

In this section, we will consider the Darboux transformation of the EZTH on the Lax operator
 \[\L=\Lambda+U+V\Lambda^{-1},\]
 i.e.
  \[\label{1darbouxL}\L^{[1]}=\Lambda+U^{[1]}+V^{[1]}\Lambda^{-1}=W\L W^{-1},\]
where $W$ is the Darboux transformation operator.
That means after Darboux transformation, the spectral problem about  $N\times N$ matrix-valued $\phi$

\[\L\phi=\Lambda\phi+U\phi+V\Lambda^{-1}\phi=\lambda\phi,\]
will become

\[\L^{[1]}\phi^{[1]}=\lambda\phi^{[1]}.\]

To keep the Lax pair of the EZTH invariant in Proposition \ref{Lax} , i.e.
   \begin{align}
\label{laxtjk}
  \partial_{t_{j}} \L^{[1]}&= [(B_{j}^{[1]})_+,\L^{[1]}],
 \partial_{s_{j}} \L^{[1]}= [(D_{j}^{[1]})_+,\L^{[1]}],\ \  B_{j}^{[1]}:=B_{j}(\L^{[1]}), D_{j}^{[1]}:=D_{j}(\L^{[1]}),\\
   \partial_{ t_{j}} \log \L^{[1]}&= [(B_{j}^{[1]})_+ ,\log \L^{[1]}],\ \ (\log \L^{[1]})_{ s_{j}}=[ -(D_{j}^{[1]})_-,\log_+ \L^{[1]} ]+
[(D_{j}^{[1]})_+ ,\log_- \L^{[1]} ],
\end{align} the dressing operator $W$ should satisfy the following dressing equation
\[W_{t_{j}}&=&-W(B_{j})_++(WB_{j}W^{-1})_+W,\ \  j\geq 0.\]
where $W_{t_{j}}$ means the derivative of $W$ by $t_{j}.$
$W$ should also satisfy the following dressing equation
\[W_{s_{j}}&=&-W(D_{j})_++(WD_{j}W^{-1})_+W,\ \  j\geq 0.\]
where $W_{s_{j}}$ means the derivative of $W$ by $s_{j}.$

Now, we will give the following important theorem which will be used to generate new solutions.

\begin{theorem}
If $\phi$ is the first wave function of the EZTH,
the Darboux transformation operator of the EZTH  \[W(\lambda)=(1-\phi(\phi(x-\epsilon))^{-1}\La^{-1})=\phi\circ(1-\La^{-1})\circ\phi^{-1},\]

will generater new solutions $U^{[1]},V^{[1]}$ from seed solutions
$U,V$

\[\label{1uN-11}U^{[1]}&=&U+(\La-1)\phi(\phi(x-\epsilon))^{-1},\\ \label{1vN-11} V^{[1]}&=&\La^{-1}V\frac{\phi\La^{-2}\phi}{\La^{-1}\phi^{2}}.\]

\end{theorem}

Define $ \phi_i=\phi_i^{[0]}:=\phi|_{\lambda=\lambda_i}$, then one can choose the specific one-fold  Darboux transformation of the EZTH as following
\[W_1(\lambda_1)=\I_N-\phi_1(\phi_1(x-\epsilon))^{-1}\La^{-1}.\]

Meanwhile, we can also get Darboux transformation on wave function $\phi$ as following

 \[\phi^{[1]}=(\I_N-\phi_1(x)(\phi_1(x-\epsilon))^{-1}\La^{-1})\phi.\]
Then using iteration on Darboux transformation, the $j$-th Darboux transformation from the $(j-1)$-th solution is as

\[\phi^{[j]}&=&(\I_N-\frac{\phi_j^{[j-1]}}{\La^{-1}\phi_j^{[j-1]}}\La^{-1})\phi^{[j-1]},\\
U^{[j]}&=&U^{[j-1]}+(\La-1)\frac{\phi_j^{[j-1]}}{\La^{-1}\phi_j^{[j-1]}},\\
V^{[j]}&=&(\La^{-1} V^{[j-1]})\frac{\phi_j^{[j-1]}}{\La^{-1}\phi_j^{[j-1]}}\frac{\La^{-2}\phi_j^{[j-1]}}{\La^{-1}\phi_j^{[j-1]}},\]
where $ \phi_i^{[j-1]}:=\phi^{[j-1]}|_{\lambda=\lambda_i},$ are wave functions corresponding to different spectrals with the $(j-1)$-th solutions $U^{[j-1]},V^{[j-1]}.$ It can be checked that $ \phi_i^{[j-1]}=0,\ \ i=1,2,\dots, j-1.$

After iteration on Darboux transformations,  we can generalize the Darboux transformation to $n$-fold case which is contained in the following theorem.

\begin{theorem}\label{ndarboux}
The $n$-fold  Darboux transformation of EZTH equation is as following
\[W_n=\I_N+t_1^{[n]}\Lambda^{-1}+t_2^{[n]}\Lambda^{-2}+\dots+t_{n}^{[n]}\Lambda^{-n}\]
where

\[ W_n\cdot\phi_{i}|_{i\leq n}=0.\]

The Darboux transformation leads to new solutions from seed solutions
\[U^{[n]}&=&U+(\La-1)t_1^{[n]},\\
V^{[n]}&=&t_n^{[n]}(x)(\La^{-n}V)t_n^{[n]-1}(x-\epsilon).\]
where
\[\notag &&(W_n)_{ij}=\frac{1}{\Delta_n}\\ \notag
&&\left|\begin{matrix}\begin{smallmatrix}
\delta_{ij}&0&\dots & \La^{-1}&\dots & 0&0&\dots & 0&\dots & \La^{-n}&\dots & 0\\
0&\phi_{1,11}(x-\epsilon)&\dots & \phi_{1,j1}(x-\epsilon)&\dots &\phi_{1,N1}(x-\epsilon)&\phi_{1,11}(x-2\epsilon)&\dots & \phi_{1,N1}(x-2\epsilon)&\dots & \phi_{1,j1}(x-n\epsilon)&\dots & \phi_{1,N1}(x-n\epsilon)\\
0&\phi_{1,12}(x-\epsilon)&\dots & \phi_{1,j2}(x-\epsilon)&\dots & \phi_{1,N2}(x-\epsilon)&\phi_{1,12}(x-2\epsilon)&\dots & \phi_{1,N2}(x-2\epsilon)&\dots & \phi_{1,j2}(x-n\epsilon)&\dots & \phi_{1,N2}(x-n\epsilon)\\
-\phi_{1,ii}(x)&\phi_{1,1i}(x-\epsilon)&\dots & \phi_{1,ji}(x-\epsilon)&\dots & \phi_{1,Ni}(x-\epsilon)&\phi_{1,1i}(x-2\epsilon)&\dots & \phi_{1,Ni}(x-2\epsilon)&\dots & \phi_{1,ji}(x-n\epsilon)&\dots & \phi_{1,Ni}(x-n\epsilon)\\
\dots&\dots&\dots&\dots & \dots&\dots &\dots&\dots &\dots&\dots & \dots&\dots &\dots\\
0&\phi_{1,1N}(x-\epsilon)&\dots & \phi_{1,jN}(x-\epsilon)&\dots & \phi_{1,NN}(x-\epsilon)&\phi_{1,1N}(x-2\epsilon)&\dots& \phi_{1,NN}(x-2\epsilon)&\dots & \phi_{1,jN}(x-n\epsilon)&\dots & \phi_{1,NN}(x-n\epsilon)\\
0&\phi_{2,11}(x-\epsilon)&\dots & \phi_{2,j1}(x-\epsilon)&\dots &
\phi_{2,N1}(x-\epsilon)&\phi_{2,21}(x-2\epsilon)&\dots & \phi_{2,N1}(x-2\epsilon)&\dots & \phi_{2,j1}(x-n\epsilon)&\dots & \phi_{2,N1}(x-n\epsilon)\\
0&\phi_{2,12}(x-\epsilon)&\dots & \phi_{2,j2}(x-\epsilon)&\dots & \phi_{2,N2}(x-\epsilon)&\phi_{2,12}(x-2\epsilon)&\dots & \phi_{2,N2}(x-2\epsilon)&\dots & \phi_{2,j2}(x-n\epsilon)&\dots  & \phi_{2,N2}(x-n\epsilon)\\
-\phi_{2,ii}(x)&\phi_{2,1i}(x-\epsilon)&\dots & \phi_{2,ji}(x-\epsilon)&\dots & \phi_{2,Ni}(x-\epsilon)&
\phi_{2,1i}(x-2\epsilon)&\dots  & \phi_{2,Ni}(x-2\epsilon)&\dots & \phi_{2,ji}(x-n\epsilon)&\dots & \phi_{2,Ni}(x-n\epsilon)\\
\dots&\dots&\dots&\dots & \dots&\dots &\dots&\dots & \dots&\dots  & \dots&\dots &\dots\\
0&\phi_{2,1N}(x-\epsilon)&\dots & \phi_{2,jN}(x-\epsilon)&\dots & \phi_{2,NN}(x-\epsilon)&\phi_{2,1N}(x-2\epsilon)&\dots & \phi_{2,NN}(x-2\epsilon)&\dots & \phi_{2,jN}(x-n\epsilon)&\dots & \phi_{2,NN}(x-n\epsilon)\\
\dots&\dots&\dots&\dots & \dots&\dots &\dots&\dots & \dots&\dots  & \dots&\dots &\dots\\
0&\phi_{n,11}(x-\epsilon)&\dots & \phi_{n,j1}(x-\epsilon)&\dots &
\phi_{n,N1}(x-\epsilon)&\phi_{n,21}(x-2\epsilon)&\dots & \phi_{n,N1}(x-2\epsilon)&\dots & \phi_{n,j1}(x-n\epsilon)&\dots & \phi_{n,N1}(x-n\epsilon)\\
0&\phi_{n,12}(x-\epsilon)&\dots & \phi_{n,j2}(x-\epsilon)&\dots & \phi_{n,N2}(x-\epsilon)&\phi_{n,12}(x-2\epsilon)&\dots & \phi_{n,N2}(x-2\epsilon)&\dots & \phi_{n,j2}(x-n\epsilon)&\dots  & \phi_{n,N2}(x-n\epsilon)\\
-\phi_{n,ii}(x)&\phi_{n,1i}(x-\epsilon)&\dots & \phi_{n,ji}(x-\epsilon)&\dots & \phi_{n,Ni}(x-\epsilon)&
\phi_{n,1i}(x-2\epsilon)&\dots  & \phi_{n,Ni}(x-2\epsilon)&\dots & \phi_{n,ji}(x-n\epsilon)&\dots & \phi_{n,Ni}(x-n\epsilon)\\
\dots&\dots&\dots&\dots & \dots&\dots &\dots&\dots & \dots&\dots  & \dots&\dots &\dots\\
0&\phi_{n,1N}(x-\epsilon)&\dots & \phi_{n,jN}(x-\epsilon)&\dots & \phi_{n,NN}(x-\epsilon)&\phi_{n,1N}(x-2\epsilon)&\dots & \phi_{n,NN}(x-2\epsilon)&\dots & \phi_{n,jN}(x-n\epsilon)&\dots & \phi_{n,NN}(x-n\epsilon)\end{smallmatrix}\end{matrix}
\right|,\]
\[\notag&&\Delta_n=\left|\begin{matrix}\begin{smallmatrix}
\phi_{1}(x-\epsilon)&\phi_{1}(x-2\epsilon)&\dots & \phi_{1}(x-n\epsilon)\\
\phi_{2}(x-\epsilon)&\phi_{2}(x-2\epsilon)&\dots & \phi_{2}(x-n\epsilon)\\
\dots&\dots&\dots & \dots\\
\phi_{n}(x-\epsilon)&\phi_{n}(x-2\epsilon)&\dots & \phi_{n}(x-n\epsilon)\\
\end{smallmatrix}\end{matrix}
\right|.\]
\end{theorem}
It can be easily checked that $W_n\phi_i=0,\ i=1,2,\dots,n.$

Taking seed solution $U=(0)_{N\times N},V=\I_N$, then using Theorem \ref{ndarboux},  one can get the $n$-th new solution of the EZTH as

\[U^{[n]}&=&(1-\La^{-1})\d_{t_{0}}\log \bar W_r(\phi_1,\phi_2,\dots\phi_n),\\
V^{[n]}&=&e^{(1-\La^{-1})(1-\La^{-1})\log \bar W_r(\phi_1,\phi_2,\dots\phi_n)},\]

where $\bar W_r(\phi_1,\phi_2,\dots\phi_n)$ is the  ``Hankel" in terms of $\Gamma$
\[\bar W_r(\phi_1,\phi_2,\dots\phi_n)=det (\La^{-j+1} \phi_{n+1-i})_{1\leq i,j\leq n}.\]
``Hankel" in terms of $\Gamma$ means in the process of calculation, we treat every element $\phi_{n+1-i}$ not directly as a matrix form but as a scalar polynomial of element $\Gamma.$ After getting the values of $U^{[n]},V^{[n]}$ in terms of $\Gamma$, we then rewrite it in matrix form.

\subsection{Soliton solutions}

After the above preparation over the first Darboux transformation, in this section, we will use the first Darboux transformation of the EZTH to generate new solutions from trivial seed solutions. In particular, some matrix-valued soliton solutions will be shown using the first Darboux transformation.

For $N=2$, one can take seed solution $U=\begin{bmatrix}0&0\\ 0&0\end{bmatrix},V=\begin{bmatrix}1&0\\ 0&1\end{bmatrix}$, then the initial wave function $\phi_i$ satisfies
\[\label{initial}\Lambda\phi+\begin{bmatrix}1&0\\ 0&1\end{bmatrix}\Lambda^{-1}\phi=\begin{bmatrix}\lambda_1&0\\\lambda_2&\lambda_1\end{bmatrix}\phi,\ 1\leq i\leq n.\]

\[\phi=\exp(\frac x{\epsilon}\log\begin{bmatrix}z_1&0\\ z_2&z_1\end{bmatrix}),\ \ z_1\neq 0,\ \ \]
with $z_1+z_1^{-1}=\lambda_1,z_2+z_1^{-1}-\frac{z_2}{z_1^2}=\lambda_2.$
\[S=E+\omega_1\La^{-1}+v\frac{x}{\epsilon}\La^{-1}-\frac{x}\omega_1{\epsilon}\La^{-2}+\dots,\ \ \omega_1=cons.\]
Under this initial equation, the operator $A_{1}$ in above Lemma \ref{modifiedLax} is in form of
\begin{align}\notag
   A_{1}&=(\Lambda+\begin{bmatrix}1&0\\ 0&1\end{bmatrix}\Lambda^{-1})\epsilon\d-(\Lambda-\begin{bmatrix}1&0\\ 0&1\end{bmatrix}\Lambda^{-1}).
    \end{align}

\[\frac{\partial \phi}{\partial s_{1}} =([(\Lambda+\begin{bmatrix}1&0\\ 0&1\end{bmatrix}\Lambda^{-1})\epsilon\d-(\Lambda-\begin{bmatrix}1&0\\ 0&1\end{bmatrix}\Lambda^{-1})]\phi,\]

Then solution $\phi$  in terms of $x,s_{1}$ can be chosen in the form
\begin{align}
\phi&=\exp(\frac{x +\lambda s_{1}}{\epsilon}\log Z+\frac{ s_{1}}{\epsilon}(-Z+Z^{-1})),\ Z=\begin{bmatrix}z_1&0\\z_2&z_1\end{bmatrix},\ \ \lambda=\begin{bmatrix}\lambda_1&0\\ \lambda_2&\lambda_1\end{bmatrix},
   \end{align}
where $Z+\begin{bmatrix}1&0\\ 0&1\end{bmatrix}Z^{-1}=\begin{bmatrix}z_1&0\\z_2&z_1\end{bmatrix}+\begin{bmatrix}1&0\\ 0&1\end{bmatrix}\begin{bmatrix}z_1^{-1}&0\\-\frac{z_2}{z_1^2}&z_1^{-1}\end{bmatrix}=\lambda=\begin{bmatrix}\lambda_1&0\\ \lambda_2&\lambda_1\end{bmatrix}.$

\[\label{uN}U^{[1]}&=&(\La-1)\frac{\cosh(\frac{x +\lambda s_{1}}{\epsilon}\log Z+\frac{ s_{1}}{\epsilon}(-Z+Z^{-1}))}{\cosh(\frac{x +\lambda s_{1}}{\epsilon}\log Z+\frac{ s_{1}}{\epsilon}(-Z+Z^{-1})-\log Z)},\\ \notag
V^{[1]}&=&(1-\Lambda^{-1})^2\log[2\cosh(\frac{x +\lambda s_{1}}{\epsilon}\log Z+\frac{ s_{1}}{\epsilon}(-Z+Z^{-1}))].\]
Using $\log Z=\begin{bmatrix}\log z_1&0\\\frac{z_2}{z_1}&\log z_1\end{bmatrix},\cosh(\begin{bmatrix}a&0\\b&a\end{bmatrix})=\begin{bmatrix}\cosh a&0\\b\cosh a&\cosh a\end{bmatrix}$, one can derive the specific elements in new solutions $U^{[1]},V^{[1]}$ as

\[\label{uN1}u_0^{[1]}&=&(\La-1)\frac{\cosh(\frac{x}{\epsilon}\log z_1 +\frac{ s_{1}}{\epsilon}(\lambda_1-z_1+z_1^{-1}))}{\cosh(\frac{x-\epsilon}{\epsilon}\log z_1 +\frac{ s_{1}}{\epsilon}(\lambda_1-z_1+z_1^{-1}))},\\
u_1^{[1]}&=&(\frac{x}{\epsilon}\frac{z_2}{z_1} +\frac{ s_{1}}{\epsilon}(\lambda_2-z_2-\frac{z_2}{z_1^2}))\frac{\cosh(\frac{x}{\epsilon}\log z_1 +\frac{ s_{1}}{\epsilon}(\lambda_1-z_1+z_1^{-1}))}{\cosh(\frac{x-\epsilon}{\epsilon}\log z_1 +\frac{ s_{1}}{\epsilon}(\lambda_1-z_1+z_1^{-1}))}\\ \notag
&&-(\frac{x-\epsilon}{\epsilon}\frac{z_2}{z_1} +\frac{ s_{1}}{\epsilon}(\lambda_2-z_2-\frac{z_2}{z_1^2}))\frac{\cosh(\frac{x}{\epsilon}\log z_1 +\frac{ s_{1}}{\epsilon}(\lambda_1-z_1+z_1^{-1}))}{\cosh^2(\frac{x-\epsilon}{\epsilon}\log z_1 +\frac{ s_{1}}{\epsilon}(\lambda_1-z_1+z_1^{-1}))},\\
v_0^{[1]}&=&(1-\Lambda^{-1})^2\log[2\cosh(\frac{x}{\epsilon}\log z_1 +\frac{ s_{1}}{\epsilon}(\lambda_1-z_1+z_1^{-1}))],\\
v_1^{[1]}&=&(1-\Lambda^{-1})^2\log(\frac{x}{\epsilon}\frac{z_2}{z_1} +\frac{ s_{1}}{\epsilon}(\lambda_2-z_2-\frac{z_2}{z_1^2})).\]

Taking $z_2=\lambda_2=u_1^{[1]}=v_1^{[1]}=0$, the above soliton solutions will be reduced to soliton solutions of the scalar-valued extended Toda chain\cite{thesis}.
\section{Bi-Hamiltonian structure and tau symmetry}

To describe the integrability of the EZTH, we will construct the Bi-Hamiltonian structure and tau symmetry of the EZTH in this section as \cite{zuo}.
For a matrix $A=(a_{ij})=\sum_{i=0}^{N-1}a_{i}\Gamma^i$, the vector field $\d_A$ over EZTH is defined by
\[\d_A=\sum_{i=0}^{N-1}\sum_{k\geq 0}a_{i}^{(k)}(\frac{\d}{\d u_{i}^{(k)}}+\frac{\d}{\d v_{i}^{(k)}}).\]

For a function $\bar f=\int f dx$, we have
\[\d_A \bar f=\int\sum_{i=0}^{N-1}\sum_{k\geq 0}a_{i}^{(k)}(\frac{\d f}{\d u_{i}^{(k)}}+\frac{\d f}{\d v_{i}^{(k)}}) dx=\int Tr_N \sum_{k\geq 0}A^{(k)}(\frac{\delta f}{\delta u^{(k)}}+\frac{\delta f}{\delta v^{(k)}}) dx,\]
where
\[(\frac{\delta}{\delta u})_{ij}=\frac{\delta}{\delta u_{ji}},\ (\frac{\delta}{\delta v})_{ij}=\frac{\delta}{\delta v_{ji}},\]
and
\[Tr_N A=the \ trace \ of \begin{bmatrix}\frac 1N&\frac 1{N-1}&\cdot &1\\ 0&\frac 1N&\cdot &\frac 12\\\\ \vdots&\vdots&\vdots&\vdots \\
 0&0&\cdot &\frac 1N\end{bmatrix} A.\]
 In this section, we will consider the EZTH on Lax operator
 \[\L=\Lambda+u+e^v\Lambda^{-1},\ \ u,v\in Z_N.\]
Then we can define the hamiltonian bracket as
\[\{\bar f,\bar g\}=\int Tr_N \sum_{w,w'}\frac{\delta f}{\delta w}\{w,w'\}\frac{\delta g}{\delta w'} dx,\ \ w,w'=u_i\ or\ v_j,\ \ 0\leq i,j\leq N-1.\]
For $u(x)=\sum_{i=0}^{N-1}u_i(x)\Gamma^i,\ \ v(x)=\sum_{i=0}^{N-1}v_i(x)\Gamma^i,$ the bi-Hamiltonian structure for the
EZTH can be given by the following two compatible Poisson brackets which is a generalization in matrix form of the extended Toda hierarchy in \cite{CDZ}

\begin{eqnarray}
&&\{v_{i}(x),v_{j}(y)\}_1=\{u_{i}(x),u_{j}(y)\}_1=0,\notag\\
&&\{u_i(x),v_j(y)\}_{1}=\frac{1}{N\epsilon}\delta_{i0}\delta_{j0} \left[\Lambda-1
\right]\delta(x-y),\label{toda-pb1}\\
&& \{u_{i}(x),u_{j}(y)\}_2={1\over N\epsilon}\left[\Lambda
e^{v(x)}-
e^{v(x)} \Lambda^{-1}\right]_{i+j} \delta(x-y),\notag\\
&& \{ u_{i}(x), v_{j}(y)\}_2 = {1\over N\epsilon}
u_{i+j}(x)\left[\Lambda-1 \right]
\delta(x-y),\label{toda-pb2}\\
&& \{ v_{i}(x), v_{j}(y)\}_2 = {1\over \epsilon}\delta_{i0}\delta_{0j} \left[
\Lambda-\Lambda^{-1}\right]\delta(x-y).\notag
\end{eqnarray}

For any difference operator $A=
\sum_k A_k \Lambda^k$, define residue $Res A=A_0$.
In the following theorem, we will prove the above Poisson structure can be as the the Hamiltonian structure of the EZTH.
\begin{theorem}
The flows of the EZTH  are Hamiltonian systems
of the form
\[
\frac{\d u_i}{\d t_{k,j}}&=&\{u_i,H_{k,j}\}_1, \  \frac{\d v_i}{\d t_{k,j}}=\{v_i,H_{k,j}\}_1,
\quad k=0,1;\ j\ge 0,
\label{td-ham}
\]
with $ t_{0,j}=t_{j},t_{1,j}=s_{j}.$
They satisfy the following bi-Hamiltonian recursion relation
\[\notag
\{\cdot,H_{1,n-1}\}_2&=&n
\{\cdot,H_{1,n}\}_1+2\{\cdot,H_{0,n-1}\}_1,\ \{\cdot,H_{0,n-1}\}_2=(n+1)
\{\cdot,H_{0,n}\}_1.
\]
Here the Hamiltonians have the form
\begin{equation}
H_{k,j}=\int h_{k,j}(u,v; u_x,v_x; \dots; \epsilon) dx,\quad k=0,1; \ j\ge 0,
\end{equation}
with
\[
 h_{0,j}&=&\frac1{(j+1)!}Tr_N Res \, \L^{j+1},\
   h_{1,j}=\frac2{j!}\,Tr_N Res\left[ \L^{j}
(\log \L-c_{j})\right].
\]

\end{theorem}

\begin{proof}
For $\beta=0$, i.e. the original Toda hierarchy, the proof is same as the proof in \cite{CDZ}.

Here we will prove that the flows $\frac{\d}{\d t_{1,n}}$ are also
Hamiltonian systems with respect to the first Poisson bracket.
In \cite{CDZ}, the following identity has been proved
\begin{equation}\label{dlgl-2} Tr_N Res\left[\L^n d (S\epsilon \d_x
S^{-1})\right] \sim Tr_N Res \L^{n-1} d \L,
\end{equation}
which show the validity of the following equivalence relation:
\begin{equation}\label{dlgl}
Tr_N Res\left(\L^n\, d \log_+ \L\right) \sim Tr_N Res\left(\L^{n-1} d \L\right).
\end{equation}
Here the equivalent relation $\sim$ is up to a $x$-derivative of
another 1-form.

In a similar way as eq.\eqref{dlgl-2}, we obtain the following equivalence relation
\begin{equation}\label{dlgl-3}
Tr_N Res\left[\L^n d (\bar S\epsilon\d_x \bar S^{-1})\right]\sim -\rm Tr_N Res \L^{n-1} d \L,
\end{equation}
i.e.
\begin{equation}\label{dlgl'}
Tr_N Res\left(\L^n\, d \log_- \L\right) \sim \rm Tr_N Res\left(\L^{n-1} d \L\right).
\end{equation}
Combining \eqref{dlgl} with \eqref{dlgl'} together can lead to
\begin{equation}\label{dlgl2}
Tr_N Res\left(\L^n\, d \log \L\right) \sim \rm Tr_N Res\left(\L^{n-1} d \L\right).
\end{equation}
Suppose
\[
A_{\alpha,n}=\sum_{k} a_{\alpha,n+1;k}\, \Lambda^k,
\]
Then from
\begin{equation}
  \label{edef3}
\frac{\partial \L}{\partial t_{k, n}} = [ (B_{k,n})_+ ,\L ]= [ -(B_{k,n})_- ,\L ], \ \ B_{0,n}=B_n,B_{1,n}=D_n,
\end{equation}
we can derive equation
\[\epsilon\frac{\partial u}{\partial t_{\beta, n}}&=&a_{\beta,n;1}(x+\epsilon)-a_{\beta,n;1}(x)\in Z_N,\ \beta=0,1,\\
\epsilon\frac{\partial v}{\partial t_{\beta, n}}&=&a_{\beta,n;0}(x-\epsilon) e^{v(x)}-a_{\beta,n;0}(x) e^{v(x+\epsilon)}\in Z_N.
\]

The equivalence relation (\ref{dlgl}) now readily follows from the above two equations.
By using (\ref{dlgl}) we obtain
\begin{eqnarray}
&&d \tilde h_{n}=\frac2{n!}\,d\,Tr_N Res\left[\L^{n}
\left(\log_+ \L-c_{n}\right) \right]
\notag\\
&& \sim \frac2{(n-1)!}\,Tr_N Res\left[\L^{n-1}
\left(\log_+ \L-c_{n}\right) d \L\right]+ \frac2{n!}\,Tr_N Res\left[\L^{n-1} d \L\right]\notag\\
&&=\frac2{(n-1)!}\,Tr_N Res\left[\L^{n-1} \left(\log_+ \L-c_{n-1}\right) d \L\right]\\
&&=Tr_N \left[a_{1,n;0}(x)du+a_{1,n;1}(x-\epsilon) e^{v(x)}dv\right].
\end{eqnarray}
It yields the following identities
\begin{equation}\label{dH1-u12}
\frac{\delta H_{1,n}}{\delta u}=a_{1,n;0}(x),\quad \frac{\delta H_{1,n}}
{\delta v}=a_{1,n;1}(x-\epsilon) e^{v(x)}.
\end{equation}
This agree with Lax equation

\[
\frac{\d u_i}{\d t_{1,n}}&=&\{u_i,H_{1,n}\}_1={1\over \epsilon} \left[
e^{\epsilon\,\d_x}-1\right]\frac{\delta H_{1,n}}
{\delta v_i}={1\over \epsilon}(a_{1,n;1}(x+\epsilon)-a_{1,n;1}(x))_i,\\
 \  \frac{\d v_i}{\d t_{1,n}}&=&\{v_i,H_{1,n}\}_1=\frac{1}{\epsilon} \left[1-e^{\epsilon\,\d_x}
\right]\frac{\delta H_{1,n}}
{\delta u_i}=\frac{1}{\epsilon} \left[a_{1,n;0}(x-\epsilon) e^{v(x)}-a_{1,n;0}(x) e^{v(x+\epsilon)}\right]_i.
\]

 From the above identities we see that
the flows $\frac{\d}{\d t_{1,n}}$ are Hamiltonian systems
of the form (\ref{td-ham}).
For the case of $\beta=1$ the recursion relation
follows from the following trivial identities
\begin{eqnarray}
&&n\, \frac{2}{n!} \L^{n} \left(\log_{\pm} \L-c_{n}\right)=\L\,
\frac{2}{(n-1)!}
\L^{n-1} \left(\log_{\pm} \L-c_{n-1}\right)-2\,\frac1{n!} \L^n\notag\\
&&=\frac{2}{(n-1)!} \L^{n-1} \left(\log_{\pm} \L-c_{n-1}\right)\,
\L-2\,\frac1{n!} \L^n.\notag
\end{eqnarray}
Then we get, for $\beta=1,$
\begin{eqnarray}
&&n a_{1,n+1;1}(x)=a_{1,n;0}(x+\epsilon)+ua_{1,n;1}(x)+e^va_{1,n;2}(x-\epsilon)-2a_{0,n+1;1}(x)\notag\\
&&=a_{1,n;0}(x)+u(x+\epsilon)a_{1,n;1}(x)+e^{v(x+2\epsilon)}a_{1,n;2}(x)-2a_{0,n+1;1}(x).\notag
\end{eqnarray}
This further leads to

\begin{eqnarray}
&&\{u_i,H_{1,n-1}\}_2=\{\left[\Lambda e^{v(x)}-e^{v(x)} \Lambda^{-1}\right] a_{1,n;0}(x)+
u(x) \left[\Lambda-1\right] a_{1,n;1}(x-\epsilon) e^{v(x)}\}_i\notag\\ \notag
&&
=n\left[a_{1,n+1;1}(x) e^{v(x+\epsilon)}-a_{1,n+1;1}(x-\epsilon) e^{v(x)}\right]_i+2\left[a_{0,n+1;0}(x) e^{v(x+\epsilon)}-a_{0,n+1;0}(x-\epsilon) e^{v(x)}\right]_i.\label{pre-recur}
\end{eqnarray}
This is exactly the recursion relation on flows for $u$. The similar recursion flow on $v$ can be similarly derived.
Theorem is proved till now.

\end{proof}

Similarly as \cite{CDZ}, the tau symmetry of the EZTH can be proved in the  following theorem.
\begin{theorem}\label{tausymmetry}
The EZTH has the following tau-symmetry property:
\begin{equation}
\frac{\d h_{\alpha,m}}{\d t_{\beta, n}}=\frac{\d
h_{\beta, n}}{\d t_{\alpha,m}},\quad \alpha,\beta=0,1,\ m,n\ge 0.
\end{equation}
\end{theorem}
\begin{proof} Let us prove the theorem for the case when $\alpha=1, \beta=0$,
other cases are proved in a similar way
\[
&&\frac{\d h_{1,m}}{\d t_{0,n}} =\frac2{m!\,(n+1)!}\,Tr_N Res[-(\L^{n+1})_-, \L^m (\log_+\L-c_m)]\notag\\
&&=\frac2{m!\,(n+1)!}\,Tr_N Res[(\L^m (\log_+\L-c_m))_+,(\L^{n+1})_-]\notag\\
&& =\frac2{m!\,(n+1)!}\,Tr_N Res[(\L^m (\log_+\L
-c_m))_+,\L^{n+1}]=\frac{\d h_{0,n}}{\d t_{1,m}}.
\]
The theorem is proved.
\end{proof}

 This property justifies the following alternative definition of another kind of
tau function for the EZTH.

\begin{definition} Another $tau$ function in $Z_N$ of the EZTH can be defined by
the following expressions in terms of the densities of the Hamiltonians:
\begin{equation}
h_{\beta,n}=\epsilon (\Lambda-1)\frac{\d\log \bar \tau}{\d t_{\beta,n}},
\quad \beta=0,1;\ n\ge 0,
\end{equation}
with $ t_{0,j}=t_{j},t_{1,j}=s_{j}.$
\end{definition}

With above two different definitions tau functions of this hierarchy, some mysterious connections between these two kinds of tau functions become an open question. One is from the Sato theory and another is from the hamiltonian tau symmetry.

\section{Conclusions and Discussions}
In this paper, we constructed a new hierarchy called EZTH and extended  Sato theory to
this hierarchy including Sato equations, matrix wave operators, Hirota
quadratic equations, the existence of
the tau function.  Similarly as extended Toda hierarchy and extended bigraded Toda hierarchy in Gromov-Witten theory of $\C P^1$ and orbiford respectively,  this
hierarchy deserves further studying and exploring because of its
potential applications in topological quantum fields  and
Gromov-Witten theory. Basing on above two different definitions tau functions of this hierarchy, to derive the  mysterious deep connections between these two kinds of tau functions with one defined from Sato theory and another from hamiltonian tau symmetry become an interesting question. This is not easy and will be included in our future work.

{\bf {Acknowledgements:}}
  Chuanzhong Li is supported by the National Natural Science Foundation of China under Grant No. 11201251, Zhejiang Provincial Natural Science Foundation of China under Grant No. LY12A01007, the Natural Science Foundation of Ningbo under Grant No. 2013A610105, 2014A610029. Jingsong He is supported by the National Natural Science Foundation of China under Grant No. 11271210, K.C.Wong Magna Fund in
Ningbo University. We would like to thank Todor E Milanov and Dafeng Zuo for their valuable discussion.
%%%%%%%%%%%%%%%%% References  %%%%%%%%%%%%%%%%%%%%%%%%%%%%%%%%%%%%%%%

\end{document}